\RequirePackage[cmex10]{amsmath}
\documentclass[graybox]{svmult}
\usepackage{helvet}         
\usepackage{courier}        
\usepackage{type1cm}        

\usepackage{makeidx}         
\usepackage{graphicx}        
\usepackage{multicol}        
\usepackage[bottom]{footmisc}
\usepackage{amssymb}

\usepackage{./mathSymbolMacros}
\usepackage{algorithmicx}
\usepackage[ruled,noend,linesnumbered]{algorithm2e}
\SetKwRepeat{Do}{do}{while}
\usepackage[noend]{algpseudocode}
\usepackage{verbatim}
\usepackage{enumerate}
\usepackage{enumitem}
\usepackage{float}
\usepackage{caption}
\usepackage[caption=false]{subfig}
\usepackage{url}
\usepackage{xargs}                      
\usepackage{xcolor}  
\usepackage[all]{xy}
\usepackage[numbers]{natbib}
\usepackage{multirow}
\usepackage{cleveref}
\usepackage[figuresright,onecolumn]{rotating}
\usepackage[]{todonotes}
\newcommandx{\tn}[2][1=]{\todo[inline,
  linecolor=LimeGreen,backgroundcolor=LimeGreen!40,bordercolor=LimeGreen,#1]{(Thakur)
    #2}}
\newcommandx{\cjm}[2][1=]{\todo[inline,
  linecolor=YellowOrangeGreen,backgroundcolor=YellowOrange!40,bordercolor=YellowOrange,#1]{(Chris)
    #2}}
\newcommandx{\ckm}[2][1=]{\todo[inline,
  linecolor=MidnightBlue,backgroundcolor=MidnightBlue!40,bordercolor=MidnightBlue,#1]{(Curtis)
    #2}}
\newcommandx{\hz}[2][1=]{\todo[inline,
  linecolor=Orchid,backgroundcolor=Orchid!40,bordercolor=Orchid,#1]{(Hao)
    #2}}
\newcommandx{\zz}[2][1=]{\todo[inline,
  linecolor=RawSienna,backgroundcolor=RawSienna!40,bordercolor=RawSienna,#1]{(Zhen)
    #2}}
\newcommandx{\all}[2][1=]{\todo[inline,
  linecolor=red,backgroundcolor=red!80,bordercolor=red,#1]{(All) #2}}

\makeindex             

\begin{document}

\title*{Approximation Techniques for Stochastic Analysis of Biological Systems}
\author{Thakur Neupane, Zhen Zhang\thanks{Corresponding author.}, Curtis Madsen, Hao Zheng, and Chris J. Myers}
\institute{Thakur Neupane \at Utah State University, Logan, UT USA,
  \email{thakur.neupane@aggiemail.usu.edu}
\and Zhen Zhang \at Utah State University, Logan, UT USA, \email{zhen.zhang@usu.edu}
\and Curtis Madsen \at Boston University, Boston, MA USA, \email{ckmadsen@bu.edu}
\and Hao Zheng \at University of South Florida, Tampa, FL USA, \email{haozheng@usf.edu}
\and Chris J. Myers \at University of Utah, Salt Lake City, UT USA, \email{myers@ece.utah.edu}
}

\authorrunning{Neupane et al.}

%
%
\maketitle

\abstract{
There has been an increasing demand for formal methods in the design
process of safety-critical synthetic genetic circuits. Probabilistic
model checking techniques have demonstrated significant potential in
analyzing the intrinsic probabilistic behaviors of complex genetic
circuit designs. However, its inability to scale limits its
applicability in practice. This chapter addresses the scalability problem by presenting a
state-space approximation method to remove unlikely states resulting
in a reduced, finite state representation of the infinite-state
continuous-time Markov chain that is amenable to probabilistic model
checking. The proposed method is evaluated on a design of a genetic toggle switch.
Comparisons with another state-of-art tool demonstrates both accuracy and efficiency of the presented method. 
}


\section{Introduction}

Computational biologists typically construct models to better understand and explore the possible behaviors of biological systems~\cite{Myers2009}.
By using formal methods, such as model checking, to analyze these models, researchers are able to ensure that certain properties hold in biological system designs~\cite{Heath2006}.
In order to numerically model check a system, the system's state space must be enumerated.
For systems that are highly concurrent and have infinite states, such as \emph{genetic circuits} (i.e., the collections of genes within DNA that interact to control the behavior of cells, see Section~\ref{sec-example} for more details), enumerating the state space can be computationally intractable due to the state space explosion problem.
Techniques such as partial order reduction that reduce the number of reachable states in a system have shown some promise in tackling this problem~\cite{Baier2004,Baier2006,Diaz2012}, but these methods often rely on transition dependencies based on the disablings (and/or enablings) and commutativity of independent transitions.
Most models of genetic circuits do not contain transitions that disable/enable other transitions leading researchers to seek other solutions to this problem.

Another way to reduce the state space of a system is to introduce threshold abstractions to collapse multiple states of the system together~\cite{Madsen2014}.
This type of abstraction works very well in systems where there exist groups of states in equivalence classes.
This is often the case in genetic circuits where firing a single transition does not have a great effect on the likelihood of firing other transitions in the system.
Although this type of abstraction has previously been successfully applied to genetic circuits, selecting the threshold values is currently done in a manual ad hoc manner.

This chapter presents an alternative method for deriving a reduced, finite state representation of a genetic circuit's behavior.
This method works by computing the approximate probability of reaching each state on-the-fly and stops exploring different paths when the cumulative path probability drops below a predetermined value, and these paths are routed to an abstract absorbing state, which accumulates probability leakage during the Markovian analysis. The resulting \emph{continuous-time Markov chain} (CTMC) can be analyzed using probabilistic model checking approaches to determine the probability that the original genetic circuit satisfies a desired temporal logic property given in \emph{continuous stochastic logic} (CSL)~\cite{Aziz2000,Kwiatkowska2007}. This chapter illustrates the utility of this method by applying it to a model of the genetic toggle switch and by comparing the results to a previous approach where the thresholds were determined by hand to produce the finite state representation~\cite{Madsen2013,Madsen2014}. Additionally, this method is compared with a state-of-art stochastic hybrid analysis tool on several benchmarks, and comparisons of results demonstrate both accuracy and efficiency of our proposed method. 


\section{Related Work}
\label{sec:sec-relatedWork}

To improve the scalability of probabilistic model checking, \emph{bisimulation minimization} (e.g., \cite{Fisler1998, Fisler1999, Fisler2002}) has been extended to the probabilistic setting~\cite{Katoen2007a} to achieve up to a logarithmic state space reduction. 
\emph{Probabilistic abstraction} (e.g., \cite{Katoen2007, Fecher2006, Hermanns2008}) applies coarser state merging to achieve better reduction, while ensuring a simulation relation between the abstract and concrete Markov models. 
A transition on the abstract Markov model has a range of probabilities, represented by an interval with the maximal and minimal probabilities for taking the transition.
In particular, \cite{Katoen2007} presents a theoretical framework for reducing \emph{discrete-time Markov chains} (DTMCs) and CTMCs using a three-valued abstraction and for model checking these abstractions.
However, how to partition the state space in this framework is not discussed, nor is the refinement of the abstractions in the case that inconclusive results are produced.
Although these reduction techniques can be powerful, they may not be effective in alleviating the exponential state growth caused by concurrency as they are not designed to tackle concurrency in the first place.
Unfortunately, concurrency is inherent in most synthetic biological systems. 


To address the state explosion problem, some approaches attempt to truncate the state space.
For instance, \cite{Mikeev2013} presents a method for selectively exploring states involving rare events; however, this technique requires the modification of parameters in the system to help guide this exploration.
Other approaches attempt to dynamically explore the state space and continually add states until the resulting state space satisfies a desired level of precision~\cite{Burrage2006,Munsky2006,Munsky2008}.

A probabilistic counter-example guided abstraction refinement approach is developed in~\cite{Hermanns2008,Wachter:2007:qest}.
Predicate abstraction is applied to programs in a probabilistic guarded command language, and counter-examples are represented as finite Markov chains, where additional predicates are extracted by using an SMT solver in the case that such counter-examples are spurious.
\cite{Kwiatkowska:2010:tacas} presents a compositional verification approach to probabilistic systems using assume-guarantee reasoning.
Both component assumptions and guarantees are represented as probabilistic safety properties.
Component verification can be expensive in this approach as it is reduced to a linear programming problem.
Furthermore, assumptions are derived manually.
Additionally, \cite{Wachter:2007:qest,Kwiatkowska:2010:tacas,Hermanns2008} are all based on probabilistic automata, which support non-determinism but with discrete-time semantics. 

In \cite{Madsen2014}, genetic circuit models are converted into CTMCs using operator site reduction abstractions relying on quasi steady-state approximations.
To avoid the state explosion problem, the authors employ a state aggregation method to collapse states together based on user provided thresholds.
While the application of probabilistic model checking to the reduced CTMC can produce results in a fraction of the time of simulation-based approaches, this method is incapable of quantifying the error introduced by this aggregation and relies on user input for good choices of thresholds.
While there has been work to address the former~\cite{Abate2015}, our method attempts to alleviate the latter by automatically determining a finite state representation by removing states that are found to be extremely unlikely during the generation of the CTMC from the genetic circuit model.


A similar approach to the one presented in this chapter is the sliding window abstraction~\cite{Henzinger2009}.
This method approximates a solution to the \emph{chemical master equation} (CME) by dividing the time period of interest into small time steps, iteratively constructing a window of an abstract state space that preserves the probability mass at the current time step, and then ``sliding'' the window in each subsequent time step to include newly generated states with significant probability while abstracting away those with negligible probability contribution until the last time step has elapsed.
This method effectively performs transient CTMC analysis on a manageable approximated state space and successively updates a state space approximation by following the direction in which probability mass moves as time evolves.
The abstract state space construction is based on a worst-case estimation of  lower and upper bounds on the populations of the chemical species.

A more recent improvement of the sliding window implementation is the \emph{STochastic Analysis of biochemical Reaction networks} (STAR) \cite{Lapin2011}. It computes approximate solutions to population Markov processes using a stochastic hybrid model that combines moment-based and state-based representations of probability distributions, and has been optimized to drop unlikely states and add likely states on-the-fly.

Our approach differs from the sliding window method in that it does not require many manual factors (e.g., several different initial states to compute a state update, a limited window size, etc.) to compute its state space.
Additionally, the method presented in this chapter has the potential to optimize the choice of the termination indicator factor to preserve accuracy while requiring a manageable state space. 
Finally, our approach is based on a reaction-based abstraction model, and as a result, is readily applied to genetic circuit models while the method in \cite{Henzinger2009} focuses on Markov chains that are specified by a finite set of transition classes.

\section{Preliminaries}
\label{sec-preliminaries}

The high-level modeling formalism used in this chapter is the
\emph{stochastic chemical kinetic} (SCK) model~\cite{Myers2009}. 
\begin{definition}
A SCK model is a tuple $\sck$ which is composed of $n$ chemical species $\speciesSet=\specSet$, $m$ reaction channels $\reactionSet=\reactSet$, and an initial molecule count of each chemical species at the beginning of analysis (i.e., $\initSt : \speciesSet^n \rightarrow \natNum$). A reaction $\reactDef{i}$ includes a \emph{propensity function} $\propen{i} \, : \natNum^{n} \, \rightarrow \posReal$ that corresponds to the probability of a reaction, and the \emph{state change vector} $\stChangeVec{i} \in
\mathbb{Z}^n$ that corresponds to the change in molecule count for each species due to reaction $\react{i}$. 
\end{definition}

A reaction $\react{i}$ can occur in state $\st \in \stSet$,
if its propensity is greater than zero (i.e., $\propen{i}(\st) > 0$). 
The propensity function $\propen{i}$ essentially determines the likelihood that $\react{i}$ occurs in the current state.  
%
After a reaction $\react{i}$ occurs, the state is updated as follows: $\stPrime = \st + \stChangeVec{i}$.
%

The execution of reactions in an SCK model creates a \emph{state graph} as defined below:

\begin{definition}
A SG is a tuple $\sgFull{}$ where
\begin{itemize}[label={--}]
\item $\stSet$ is a non-empty set of states,  
\item $\TR \subseteq \stSet \times \reactionSet \times \stSet$ is the set of state transitions. 
\item $\initSt : \speciesSet^n \rightarrow \natNum$ is the initial state.
\end{itemize}
Note that $|\sg{}|$ represents the state count of $\sg{}$.
\end{definition} 

For most SCK models of real biological networks, they incur an infinite number of states.
Therefore, the goal of this chapter is to find a finite subset of the
states that sufficiently represents states that are actually likely to
occur. Once a finite state graph is obtained, properties can be verified on this state graph using probabilistic model checking.


\emph{Probabilistic model checking} is a formal verification method
for checking quantitative properties of probabilistic systems. The
models of interest include DTMCs and CTMCs, both of which belong to a
class of stochastic processes that are used to reason about random
phenomena in application domains such as synthetic biology. Both Markov models are essentially a transition system with each transition labeled by a discrete probability for DTMCs or a transition rate for CTMCs.
A DTMC is a transition system with a discrete probability labeled on
each transition~\cite{Kwiatkowska2007}, which describes the likelihood
of a single step moving from one state to another. A CTMC, on the
other hand, is a transition system with a transition rate $r(s, s')$
labeled on the transition emanating from state $s$ to $s'$. This rate
determines the probability of executing this transition within $t$
time units, which is $e^{-r(s,s') t}$.
The rate $r(s, s')$ uniquely
characterizes an exponential distribution to govern the average state
residence time of state $s$, which is $\frac{1}{r(s, s')}$. CTMCs allow for modeling of
real-time systems, as the delay of a transition can be any arbitrary
real value.

Properties to verify using probabilistic model checking are specified
using \emph{Probabilistic Computation Tree Logic}
(PCTL)~\cite{Hansson1994} for DTMCs and CSL for CTMCs. PCTL extends
\emph{Computation Tree Logic} (CTL)~\cite{Clarke1986} by replacing
existential and universal path quantifiers with a probability
operator, and hence expresses probabilistic properties for a DTMC. 
In addition to path probabilities, two traditional properties of CTMCs are the \emph{transient} and \emph{steady-state} behaviors. Transient analysis reports the probability of being in each state of the Markov chain at a particular time instant, and steady-state analysis gives the corresponding probability in the long-run. 
Model checking algorithms for PCTL
(e.g. \cite{Courcoubetis1988,Courcoubetis1995,Hansson1994}) have
identical structure to the model checking algorithm for CTL. Model checking CTMC
first discretizes the CTMC into an \emph{embedded} DTMC, from which
many properties of the corresponding CTMC can be deduced, for example,
checking state reachability properties regardless of how long it
takes, and the expected time objectives. For checking state
reachability within some time bound, the CTMC is discretized into a
\emph{uniformized} DTMC with the iterative numerical method
\emph{uniformization}~\cite{Grassmann1977, Gross1984}. The uniformized
DTMC preserves the state resident time so that its transient behavior
is equal (up to some accuracy) to the corresponding CTMC.

In order to perform probabilistic model checking on CTMCs, CSL can be used.
CSL properties consist of state formulae (formulae that are either true or false in a specific state) and path formulae (formulae that are either true or false along a specific path).
CSL properties are specified using the following grammar:
\begin{eqnarray*}
\mathit{Prop} & ::= & {\tt U}(\textup{T},\Psi,\Psi) \: | \: {\tt F}(\textup{T},\Psi) \: | \: {\tt G}(\textup{T},\Psi) \: | \: {\tt St}(\Psi) \\
\Psi & ::= & {\tt true} \: | \: \Psi \wedge \Psi \: | \: \neg \Psi \: | \: \phi \geqslant \phi \: | \: \phi > \phi \: | \: \phi = \phi \\
\phi & ::= & v_i \: | \: c_i \: | \: \phi + \phi \: | \: \phi - \phi \: | \: \phi * \phi \: | \: \phi / \phi \: | \: \mathit{Prop} \\
\textup{T} & ::= & {\tt true} \: | \: \textup{T} \wedge \textup{T} \: | \: \neg \textup{T} \: | \: t \geqslant c_i \: | \: t > c_i \: | \: t = c_i
\end{eqnarray*}
where $v_i$ is a variable, $c_i$ is a constant, and $t$ stands for time in the system.
In CSL, $\Psi$ is a state formula that can be either comparisons between numerical expressions, $\phi$, or other state formula combined using logical connectives.
A CSL property, $\mathit{Prop}$, is a path property over state formula.  For example, the \emph{Until property} is of the form 
${\tt U}(T,\Psi_1,\Psi_2)$, and it returns the probability that along paths originating in the current state, $\Psi_1$ remains true until $\Psi_2$ becomes true during the time specified by time expression, $T$. The eventually operator, ${\tt F}$, is simply a shorthand for an until property where $\Psi_1$ is {\tt true}.  The globally true operator, ${\tt G}$, is another shorthand that specifies that $\Psi$ remains true during the time in which $T$ evaluates to true.  The steady-state operator, ${\tt St}$, returns the probability that when the SCK model reaches a steady state that it has reached a state where $\Psi$ is true.  Finally, CSL formulae, $\mathit{Prop}$, can be nested within other formula, creating recursive properties.

For example, the CSL property ${\tt St}(x > 5 \: \wedge \: y \geqslant 10)$ would return the probability that in the steady state, the system reaches a state where the variable $x$ is greater than $5$ and the variable $y$ is greater than or equal to $10$.
Alternatively, the CSL property ${\tt F}(t > 100 \wedge \neg(t \geqslant 200), x > 5 \: \wedge \: y \geqslant 10)$ would return the probability that the system follows an execution path originating in the initial state where the variable $x$ becomes greater than $5$ and the variable $y$ becomes greater than or equal to $10$ sometime between $100$ and $200$ time units non-inclusive.
For a path to satisfy this property, the system does not need these conditions to hold true for the entire 100 time unit interval; they just both need to become true simultaneously at some point within this time frame.


\section{Motivating Example}
\label{sec-example}

A genetic circuit is constructed using DNA, and it typically includes, at a minimum, regions that act as \emph{promoters}, \emph{ribosome binding sites} (RBS), \emph{coding sequences} (CDS), and \emph{terminators}.  The promoters are regions where \emph{transcription} is initiated when an \emph{RNA polymerase} (RNAP) molecule binds, and then begins to walk along the DNA copying the sequence to form a \emph{messenger RNA} (mRNA) molecule until it reaches the location of the terminator.  The terminator causes the RNAP to be released and thus ends transcription.  The RBS region when copied to an mRNA results in a region that binds to a \emph{ribosome} to initiate the \emph{translation} process.  During translation, the CDS region on the mRNA is used as instructions following the \emph{genetic code} to select the \emph{amino acids} to use to construct a \emph{protein}. Proteins are a fundamental component for almost all molecular functions within a cell. Proteins can also bind to promoters to \emph{activate} or \emph{repress} transcription, i.e., increasing or decreasing the associated promoter’s binding affinity to RNAP. 

The motivating example used in this chapter is a genetic circuit for a toggle switch~\cite{Gardner2000} shown in Figure~\ref{fig-toggle}.  This genetic circuit is constructed from two \emph{transcriptional units}.  The one on the left begins with the promoter $P_{tet}$ (shown as a bent arrow), followed by its RBS (shown as a half circle), a CDS that codes for the protein LacI, and finally a terminator (shown as a $\top$).  The one on the right begins with the $P_{lac}$ promoter, which initiates transcription of the CDSs for the TetR protein and the \emph{green fluorescent protein} (GFP).  GFP is a reporter, since the cells glow green when it is present.  The switch like behavior is created by mutual repression.  Namely, the TetR protein binds to $P_{tet}$ to repress LacI production, while the LacI protein binds to $P_{lac}$ to repress TetR production.  The state of the switch is changed by adding small molecule \emph{chemical inducers}.   Namely, when the switch is OFF (i.e., LacI is present but no TetR or GFP is present), IPTG can be added, which binds to LacI forming the complex C1, which is unable to repress $P_{lac}$.  This situation leads to TetR and GFP being produced, which represses LacI production and thus changes the switch to the ON state.  To change back to the OFF state, aTc can be added, which binds to TetR to form the complex C2, which is unable to repress $P_{tet}$. This situation leads to LacI being produced, which represses further production of TetR and GFP and thus the changes the genetic toggle switch to the OFF state.

\begin{figure}[tbhp]
  \centering
  \includegraphics[width=\columnwidth]{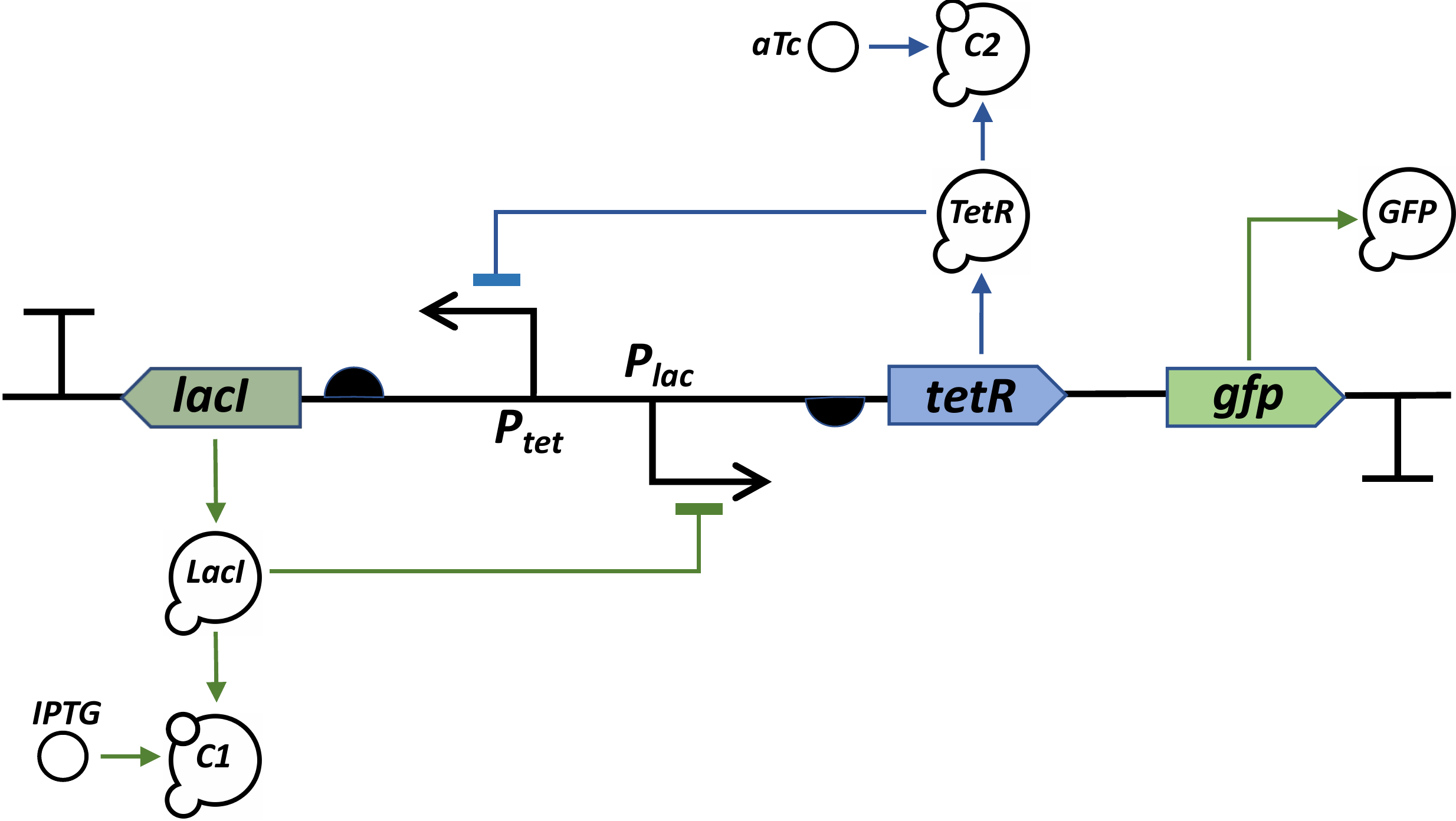}
   \caption{The genetic toggle switch. This switch is created using two repressors, LacI and TetR, which repress each others production, denoted by the $\bot$ and $\top$ arrows on
     promoters $P_{tet}$ and $P_{lac}$. The small molecule IPTG can bind to LacI, effectively reducing LacI's ability to repress TetR and GFP production.  Similarly, the small molecule aTc can bind to TetR to reduce TetR's ability to repress LacI's production. To indicate the ON and OFF states of this switch, this circuit includes the reporter protein GFP to cause the cell to glow green when it is present. 
     }
\label{fig-toggle}
\end{figure}

One possible reaction-based model of the genetic toggle switch is shown in Figure~\ref{fig-AbsChemReactsToggle}.  This model is derived from a more detailed model, using quasi-steady-state approximations and reaction-based abstractions as described in~\cite{Kuwahara2006,Myers2009}.  This model is composed of a species for each protein (i.e., LacI, TetR, and GFP) and each small molecule (i.e., IPTG and aTc).  This model also includes a production reaction for each promoter, $P_{tet}$ and $P_{lac}$, and a degradation reaction for each protein.  The reactions are shown as boxes in the diagram, with their propensity functions shown inside the boxes. The parameters for these propensity functions are given in Table~\ref{tbl-params}. Note that these are simply reasonable default parameters and not measured experimentally, and they can be easily updated if better information becomes available. The edges are labeled to indicate \emph{reactants} (r), species consumed by the reactions, \emph{products} (p), species produced by the reactions, and \emph{modifiers} (m), species neither produced or consumed.  The \emph{stoichiometry}, the number of molecules produced or consumed, for each reaction is assumed to be 1, unless indicated otherwise (e.g., production reactions produce $np$ molecules).  The degradation reactions have a propensity that is just the degradation rate, $k_d$, times the current number of molecules of the species that is degrading.  The production reactions have a propensity that is the number of molecules produced, $np$, times the rate of production, $k_p$, times the proportion of promoters bound to RNAP in steady-state.  This proportion is a function of the amount of repressor molecules present in free form (i.e., not bound to the corresponding small molecule inducer).  Further details are outside the scope of this chapter, but they can be found in~\cite{Kuwahara2006,Myers2009}.  

\begin{figure}[tbhp]
\begin{center}
\[
\scalebox{1.0}{
\xymatrix{
& *+[F]{\scriptstyle k_{d} \cospecies{\species{LacI}}} \ar@{{<}-}[d]^{r} & \\
\species{IPTG} \ar@{-}[d]^{m} & \species{LacI} \ar@{-}[dl]^{m} & \\
*+[F]{\frac{n_p k_p \cospecies{\species{P_{lac}}} K_o \cospecies{\species{RN\!AP}}} {1 + K_o \cospecies{\species{RN\!AP}} + \left(K_r \frac{\cospecies{\species{LacI}}}{1+K_c \cospecies{\species{IPTG}}}\right)^{n_c}}} \ar@{-{>}}[d]^{np,p} \ar@{-{>}}[dr]^{np,p} & & *+[F]{\frac{n_p k_p \cospecies{\species{P_{tet}}} K_o \cospecies{\species{RN\!AP}}} {1 + K_o \cospecies{\species{RN\!AP}} + \left(K_r \frac{\cospecies{\species{TetR}}}{1+K_c \cospecies{\species{aTc}}}\right)^{n_c}}} \ar@{-{>}}[ul]^{np,p} \\
\species{GFP} & \species{TetR} \ar@{-}[ur]^{m} & \species{aTc} \ar@{-}[u]^{m} \\
*+[F]{\scriptstyle k_{d} \cospecies{\species{GFP}}} \ar@{{<}-}[u]^{r} & *+[F]{\scriptstyle k_{d} \cospecies{\species{TetR}}} \ar@{{<}-}[u]^{r} &
}
}
\]
\caption[Reaction graph for the genetic toggle switch after applying reaction-based abstractions to the chemical reaction network.]{Reaction graph adapted from~\cite{Madsen2013} for the genetic toggle switch after applying reaction-based abstractions to the chemical reaction network.  
}
\label{fig-AbsChemReactsToggle}
\end{center}
\end{figure}
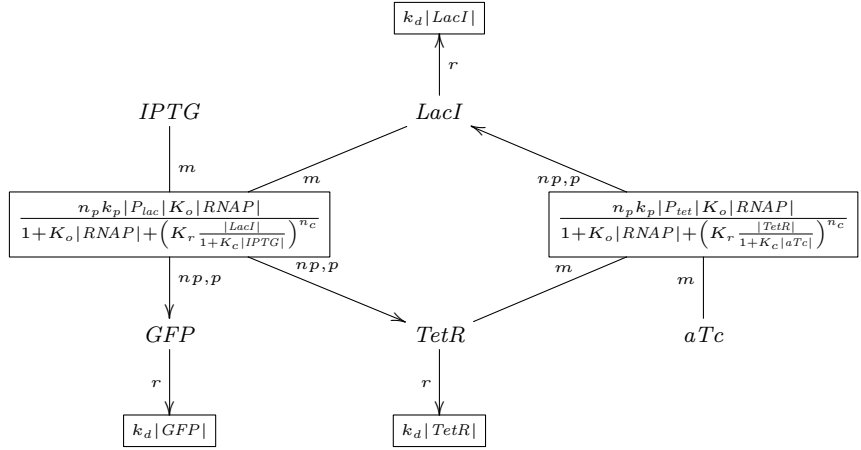

\begin{table}[tbp]
\begin{center}
\caption{List of parameters for the genetic toggle switch model.}
\vspace{2mm}
\begin{tabular}{|@{$\:\:\:\:$}l@{$\:\:\:\:$}|@{$\:\:\:\:$}c@{$\:\:\:\:$}|@{$\:\:\:\:$}c@{$\:\:\:\:$}|@{$\:\:\:\:$}c@{$\:\:\:\:$}|}
\hline
Parameter&Symbol&Value&Units\\
\hline
Degradation rate&$k_d$&0.0075& sec$^{-1}$\\
\hline
Complex formation equilibrium&$K_c$&0.05&molecule$^{-1}$\\
\hline
Stoichiometry of binding&$n_c$&2&molecules\\
\hline
Repression binding equilibrium&$K_r$&0.5&molecule$^{-1}$\\
\hline
RNAP binding equilibrium&$K_o$&0.033&molecule$^{-1}$\\
\hline
Open complex production rate&$k_p$&0.05&sec$^{-1}$\\
\hline
Stoichiometry of production&$np$&10&dimensionless\\
\hline
Number of RNAP molecules & $\cospecies{\species{RN\!AP}}$ & 30 & molecules \\
\hline
Number of $P_{tet}$ promoters & $\cospecies{\species{P_{tet}}}$ & 2 & molecules \\
\hline
Number of $P_{lac}$ promoters & $\cospecies{\species{P_{lac}}}$ & 2 & molecules \\
\hline
\end{tabular}
\label{tbl-params}
\end{center}
\end{table}

Unlike an electronic circuit, the behavior of a genetic toggle switch circuit is extremely noisy due to the small molecule counts involved. It is, therefore, necessary to evaluate a genetic circuit's behaviors using stochastic analyses. Figure~\ref{fig-toggleSSA} shows the average output response of 100 stochastic simulation runs using Gillespie’s \emph{stochastic simulation algorithm} (SSA)~\cite{Gillespie1977}. These simulations start with the same initial state with 60 LacI molecules, and 0 for other species. At time 5,000, 100 molecules of IPTG are applied, which activates the production of TetR and GFP to bring them to the high state, and represses LacI to slow down its production to allow its degradation to reduce its molecule count. When the input IPTG is removed at time 10,000 making both inputs absent, the outputs retain their current states. At time 15,000, applying inducer aTc causes the circuit to switch output states again. Removing aTc at time 20,000, once again leaves the outputs to hold their states. It should be noted that this figure shows the average output responses of 100 simulation runs, as an individual run may fail to exhibit meaningful logical behavior due to the noise in the circuit. This chapter aims to efficiently determine the probability of erroneous behavior induced by the inherent noisy nature of genetic circuits. 

\begin{sidewaysfigure}
\begin{center}  
  \includegraphics[width=0.95\textheight]{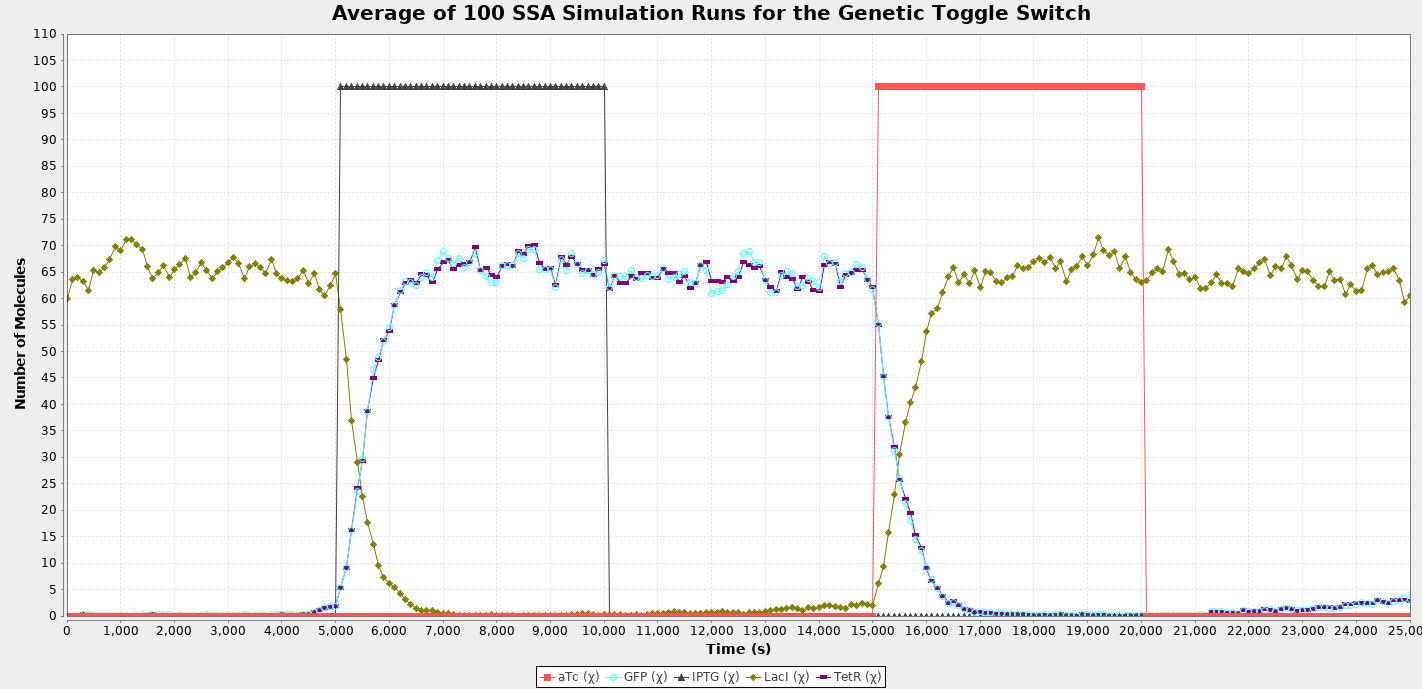}
   \parbox{0.95\textheight}{\caption{\label{fig-toggleSSA}The average of 100 stochastic simulation runs of the genetic toggle switch circuit. The LacI  molecule count drops to low soon after the introduction of IPTG at time 5,000s, which is then removed at time 10,000s. Its molecule count sharply rises to high when aTc is added at time 15,000s, which is then removed at time 20,000s, leaving LacI to stay at a high molecule count.
TetR has the opposite behavior, which is closely followed by GFP.}}
\end{center}  
\end{sidewaysfigure}

\section{State Space Approximation and Analysis} 
\label{sec-methods}

\Cref{alg-top,alg-bfs,alg-add-absorbing-state} describe the state space
approximation procedures for a given SCK model $\sck$ with
reaction-based abstractions. Note that models with reaction-based
abstractions utilize quasi steady-state
approximations~\cite{Rao2003}, where extremely fast reactions are
approximated as parameters on propensity functions to prevent
starvation of other slower reactions during stochastic analysis. The presented state
space approximation method assumes that probability mass is
distributed on a finite and relatively small number of states, and the
probability mass does not distribute uniformly as time progresses. 

With a given SCK model $\sck$, state
space generation starts by assigning the sole initial state
$\initSt$ a $1$ to the termination indicator $\stTermCurShort$, as shown in
Algorithm~\ref{alg-top}. The termination indicator is a
function $\stTermCurShort \, : \, \stSet \rightarrow \posReal$, which
indicates whether state search should terminate from
a state onwards. The initial state graph $\sg{0}$ includes the initial
state $\initSt$ as its set of states. The subsequent state graphs are then constructed and refined by
Algorithm~\ref{alg-bfs}. In general, both
state graphs $\sg{k-1}$ and $\sg{k}$ are constructed based on the same
SCK model $\sckShort$ and refined from the same initial state
$\initSt$. The difference is that $\sg{k}$ refines $\stTermCurShort$
values for some explored states in $\sg{k-1}$, which may expand
$\sg{k-1}$ to include new states in $\sg{k}$. This process of expansion
and refinement is repeated until the size of the approximate state
graph stabilizes, at which point an absorbing state is added to this
state graph by
Algorithm~\ref{alg-add-absorbing-state}. Algorithm~\ref{alg-top}
terminates by returning the approximated state graph $\sg{k}$. 

\begin{algorithm}[tbh]
\caption{Construction of approximate state space.}
\label{alg-top}
\KwIn{An SCK model $\sck$.} 
\KwOut {Approximated state graph $\sgFull{k}$.}
\BlankLine
$\sgFull{0}$, where $\stSet^{0} = \{ \initSt \}, \TR^{0} = \emptyset$\;
$\stTermCur{\initSt} := 1, \stTermNext{\initSt} := 0$\; 
$k := 0$\;

\Repeat{$|\sg{k}| = |\sg{k-1}|$} {
$k := k+1$\;
Construct finite state graph $\sgFull{k}$ for $\sckShort$ using Algorithm~\ref{alg-bfs}.\label{alg-top-call_bfs}\\
}

Update $\sg{k}$ by adding an an extra absorbing state $\stAbsorb$
using Algorithm~\ref{alg-add-absorbing-state}.
\end{algorithm}

Algorithm~\ref{alg-bfs} constructs the approximate finite state space
based on a user-defined termination indicator $\stTermToler$. Starting with
the initial state $\initSt$, all possible reactions are scheduled to be
explored (line~\ref{alg-bfs-possibleReact}). For each such reaction
$\react{i}$, its updated state $\stPrime$ is obtained by adding
the state-change vector $\stChangeVec{i}$ specified by $\react{i}$ to the current state
$\st$ (line~\ref{alg-bfs-stUpdate}). It should be noted that since the
state search in each iteration $k$ begins at the same initial state $\initSt$,
$\stPrime$ may not necessarily be a new state after this step. The termination indicator value at
the current state $\st$ is then compared against $\stTermToler$ to determine
if state exploration should continue (line~\ref{alg-bfs-determTerm}). If the former is lower, then
$\st$ becomes a (partially) terminal state, whose state expansion only
includes its outgoing transitions leading to existing states in
the current state set $\stSet^{k}$, but omits transitions leading to states not in $\stSet^{k}$. 
Therefore, if $\stPrime$ already exists in $\stSet^{k}$ (line \ref{alg-bfs-fixNxtSt}), the algorithm includes the new state-transition relation $(\st, \react{i}, \stPrime)$ and updates its
termination indicator (lines~\ref{alg-bfs-stTranUpdate} to
\ref{alg-bfs-update}). For every state $\stPrime$ to be updated, its
predecessor set is constructed (line~\ref{alg-bfs-pred}). Each element of this set is a pair of
the predecessor state $\st$ and the reaction index $i$, in which a unique
existing state transition $(\st, \react{i}, \stPrime)$ defines reachability of
$\stPrime$ from $\st$ through reaction $\react{i}$. Then the updated termination
indicator $\stTermNext{\stPrime}$ is determined by
line~\ref{alg-bfs-update}. It should be noted that
the updated termination indicator $\stTermNextShort$ is not used to update termination
indicator values for other states explored in the current iteration
$k$, and only becomes available at the end of the current iteration,
at which point it is assigned to the current termination indicator
$\stTermCurShort$ (line~\ref{alg-bfs-stTermUpdate}). For each
predecessor state $\st$ of $\stPrime$, its contribution to
$\stTermNext{\stPrime}$ is the product of its current state
termination value $\stTermCur{\st}$ and the probability of
transitioning from $\st$ to $\stPrime$, defined as the ratio of
propensity $\propen{i}$, evaluated at $\st$, to the sum of all
propensities at this state. Intuitively, $\stTermNext{\stPrime}$
accumulates path probabilities from all of its predecessor states that
have been explored in iteration $k$. On line~\ref{alg-bfs-explored}, the function 
$\texttt{explored}(\stPrime, k)$ checks whether state $\stPrime$ has
been either expanded or updated at the current iteration $k$. This
state can be a state discovered at the current iteration or any of the previous iterations.
This state is only scheduled to be explored, if it has not been explored
yet in the current iteration. 
For the case where $\stTermCur{\st} \geqslant \stTermToler$,
in addition to updating the state-transition 
relation and termination indicator for $\stPrime$, the algorithm
includes it in the state set $\stSet^{k}$ (line~\ref{alg-bfs-addToStSet}). This is
because $\st$ cannot be the terminal state due to its large
termination indicator value, and therefore its successor
$\stPrime$ becomes the potential candidate for a terminal state. This state
is then scheduled for exploration, if the current iteration has not
explored it. 


The termination indicator update is performed every time
a new incoming path is added to a state. It is crucial to have
frequent updates since a new incoming path can add its
probability contribution to the state, potentially bringing the
termination indicator value above $\stTermToler$, which in turn
changes a terminal state to be non-terminal. This update, therefore,
guarantees to explore a state with many incoming paths whose
accumulative probabilities are significant, although each individual
one might be low compared to $\stTermToler$.

\begin{algorithm}[tbh]
\caption{State space construction and approximation using breadth-first search.}
\label{alg-bfs}
\KwIn {An approximated global state graph $\sgFull{k-1}$.}
\KwOut {Updated state graph $\sgFull{k}$.}
\BlankLine
$\stSet^{k} := \stSet^{k-1}$\;
${\TR}^{k} := {\TR}^{k-1}$\;
$Enqueue(queue, \initSt)$\;
\While{$queue \neq \emptyset$} {
  $\st := Dequeue(queue)$\;
  \ForAll {$i \in \{ j \; | \; \propen{j}(\st) > 0 \}$} { \label{alg-bfs-possibleReact}
    Determine the state after reaction $\react{i}$: $\stPrime :=\st +
    \stChangeVec{i}$\label{alg-bfs-stUpdate}\;
    \If {$\stPrime \notin \stSet^{k}$} {
      $\stTermCur{\stPrime} := 0$\; \label{alg-bfs-init}
      $\stTermNext{\stPrime} := 0$\;}
    \If {$\stTermCur{\st} < \stTermToler$} { \label{alg-bfs-determTerm}
    \If {$\stPrime \in \stSet^{k}$} { \label{alg-bfs-fixNxtSt}
          ${\TR}^{k} := {\TR}^{k} \cup \{(\st, \react{i},
          \stPrime)\}$\; \label{alg-bfs-stTranUpdate}
          $\predStSet{\stPrime} :=  \{(\st, i) \; | \; (\st, \react{i},
             \stPrime) \in \TR^{k}, \forall i \in (1, \cdots, m)
             \}$\label{alg-bfs-pred}\;
          $\stTermNext{\stPrime} := \sum_{(\st, i) \in \predStSet{\stPrime}} \left(\stTermCur{\st} \cdot
            \frac{\propen{i}(\st)}{\sum_{j=1}^{m}
              \propen{j}(\st)}\right)$\label{alg-bfs-update}\;
          \If {$\lnot \texttt{explored}(\stPrime, k)$} { \label{alg-bfs-explored}
            $Enqueue(queue, \stPrime)$\;
          }
        }
    }
    \Else {
      ${\TR}^{k} := {\TR}^{k} \cup \{(\st, \react{i}, \stPrime)\}$\;
      $\predStSet{\stPrime} :=  \{(\st, i) \; | \; (\st, \react{i},
      \stPrime) \in \TR^{k}, \forall i \in (1, \cdots, m)
      \}$\;
      $\stTermNext{\stPrime} := \sum_{(\st, i) \in \predStSet{\stPrime}}
      \left(\stTermCur{\st} \cdot\frac{\propen{i}(\st)}{\sum_{j=1}^{m}
        \propen{j}(\st)}\right)$\label{alg-bfs-update2}\;
      $\stSet^{k} := \stSet^{k} \cup \{\stPrime\}$\label{alg-bfs-addToStSet}\; 
      \If {$\lnot \texttt{explored}(\stPrime, k)$} {
      	$Enqueue(queue, \stPrime)$\;
     }
   } \label{alg-bfs-1step}
 }
}
\BlankLine
\ForAll {$\st \in \stSet^{k}$}{
  $\stTermCur{\st} := \stTermNext{\st}$\label{alg-bfs-stTermUpdate}\;
}

\end{algorithm}

The theoretical state space for the genetic toggle switch described in
Section~\ref{sec-example} is infinite. To analyze the model, the
state space is truncated based on the value of $\stTermToler$. This
truncation, however, leads to probability leakage (i.e., cumulative probabilities of
reaching states not included in the explored state space) during the CTMC analysis. 
To account for probability loss, an absorbing state $\stAbsorb$ is created as the sole successor
state for all terminal states on each truncated path, and is added by
Algorithm~\ref{alg-add-absorbing-state} to the state space 
generated by Algorithm~\ref{alg-top}. For all states
in the global state set, each possible reaction for state $\st$ are
checked for exploration. For each reaction $\react{i}$, if it has not
been explored, its updated state $\stPrime$ is set to $\stAbsorb$
(line~\ref{alg-add-absorbing-state-updateTran}). It is obvious that
all unexplored transitions from such a terminal state $\st$ lead to 
the absorbing state. 

\begin{algorithm}[tbh]
\caption{Absorbing state update from approximated global state graph.}
\label{alg-add-absorbing-state}
\KwIn {An approximated global state graph $\sg{}$.}
\KwOut {Updated state graph $\sg{}$ with an absorbing state $\stAbsorb$.}
\BlankLine
  
  ${\stSet} := {\stSet} \cup \{\stAbsorb\} $\;
  \ForAll {$\st \in \stSet$} {
    \ForAll {$i \in \{ j \; | \; \propen{j}(\st) > 0 \}$} {
      Determine the state after reaction $\react{i}$: $\stPrime :=\st + \stChangeVec{i}$ \;
      \If {$(\st, \react{i}, \stPrime) \notin \TR$} { \label{alg-add-absorbing-state-updateTran}
        ${\TR} := {\TR} \cup \{(\st, \react{i}, \stAbsorb)\}$\;
      }
    }
  }
\end{algorithm}

The state graph returned by Algorithm~\ref{alg-top} is
essentially a (sparse) representation of the transition rate matrix. A
standard CTMC analysis can be applied directly to it to compute the approximate
probability distribution. It should be noted that the termination
indicator value for each state is only used to determine terminal
states, and is omitted for the CTMC analysis. 

With the addition of the absorbing state, the CTMC analysis provides a
probability bound $[l, u]$, where $0 \leqslant l < u \leqslant 1$, and
$(u - l)$ is the probability accumulated in $\stAbsorb$. Assuming the
actual probability to satisfy a CSL property $\phi$ is $p$, then it holds
that $l \leqslant p \leqslant u$. Because the lower bound $l$ does not account for probabilities from paths that, if were not truncated, would feed probabilities back to the explored states, as is the case for calculating $p$. For the upper bound $u$, it is always greater or equal to $p$. Because $u$ includes probabilities accumulated by the absorbing state, of which probabilities from truncated paths that would lead to falsification of $\phi$ are counted, in addition to probabilities of those leading to the satisfaction of $\phi$.

{\bf Complexity}: The size of generated state space models depends on the distribution of probability over states and the termination threshold.  Therefore, detailed characterization of state space complexity is challenging.  Intuitively, the state space complexity increases as the termination threshold decreases.  This is because  a lower termination threshold would allow exploration of states with lower accumulated path probabilities that would otherwise be ignored with a higher termination threshold.  Exploration of these states would likely lead to other new states.  Moreover, if the majority of probability is distributed over a small number of states, a smaller number of states may be explored compared with a more even probability distribution.

The complexity is highly dependent on the user provided termination indicator $\stTermToler$. Determining a reasonable value of $\stTermToler$ can be an iterative process. Initially, $\stTermToler$ can be set to any value that satisfies 
$0 < \stTermToler << 1$, and a state graph and probability bound 
$[l, u]$ can be generated. The user can then decrease the value of $\stTermToler$, if necessary, to tighten the probability bound window. The user can repeat the process until the probability bound returned is guaranteed to prove or disprove the given CSL property.

\section{Proof of the Termination Condition}
The presented algorithms in Section~\ref{sec-methods} are guaranteed to terminate under certain conditions. This section provides a description of the termination conditions for each algorithm, and presents a proof for termination. 

To facilitate the following proof, we first define finite paths of a state graph and depth for breadth-first search. A finite path $\rho$ of a state graph is a sequence $\st_0 \xrightarrow{\react{0}} \st_1 \xrightarrow{\react{1}} \ldots \st_{n-1} \xrightarrow{\react{n-1}} \st_n$ such that for every ${0 \leqslant i < n}$, $(\st_i, \react{i}, \st_{i+1}) \in \TR$ holds for some $\react{i}$. State $\st_n$ is reachable in $\sg{}$ if $\st_n$ is reachable from the initial state through a finite path included in $\sg{}$. Denote the set of all states with depth $\imath$ as $\stSetDepth{\imath}{}$. At depth 0, $\stSetDepth{0}{} = \{ \initSt \}$. At depth $\imath > 0$, $\stSetDepth{\imath}{}$ is obtained by collecting all newly created states resulted from the one-step BFS search on all states in $\stSetDepth{\imath - 1}{}$. Therefore, the depth for a state is determined when it is explored for the first time. Note that $\stSetDepth{0}{} \cap \stSetDepth{1}{} \cdots \stSetDepth{\imath-1}{} \cap \stSetDepth{\imath}{} = \emptyset$. 

Termination condition for Algorithm~\ref{alg-top} requires that, as both the depth $\imath$ and iteration $k$ increase, the sum of termination indicator values for all states of $\stSetDepth{\imath}{k}$ decreases, with possibly finitely many iterations where this sum remains constant. This is formulated by Theorem~\ref{theorem-termination} below. 

\begin{theorem}[Termination of Algorithm~\ref{alg-top}] Algorithm~\ref{alg-top} terminates after a finite number of iterations with a given $\stTermToler$, where $0< \stTermToler << 1$, if the state graph $\sg{\jmath+1}$ satisfies the following condition: for each depth $\jmath > 0$, there must exist depth $0 \leqslant \imath \leqslant \jmath$ such that $\st_d \xrightarrow{\react{h}} \st_{d+1} \xrightarrow{\react{i}} \ldots \st_{d+(m-1)} \xrightarrow{\react{l}} \st_{d+m}$ is a finite path in $\sg{\jmath+1}$, and $\st_d \in \stSetDepth{\imath}{\jmath + 1}$, $\st_{d+(m-1)} \in \stSetDepth{\jmath}{\jmath+1}$, $\st_{d+m} \in \stSetDepth{0}{\jmath+1} \cup \stSetDepth{1}{\jmath+1} \cup \cdots \stSetDepth{\jmath - 1}{\jmath +1} \cup  \stSetDepth{\jmath}{\jmath+1}$, and $m \in \mathbb{Z}_{\geqslant 0}$.
\label{theorem-termination}
\end{theorem}

\begin{proof}
  Initially, $\stSetDepth{}{0} = \{\initSt\}$ and $\stTermCur{\initSt} = 1$. At iteration $k=1$, during the construction of $\sg{1}$ (line~\ref{alg-top-call_bfs} of Algorithm~\ref{alg-top}), each state at depth 1, $\stDepth{1} \in \stSetDepth{1}{1}$ , is discovered for the first time when $\initSt$ is explored (line~\ref{alg-bfs-possibleReact} to~\ref{alg-bfs-1step} of Algorithm~\ref{alg-bfs}). Therefore, the current termination indicator $\stTermCur{\stDepth{1}}$ is assigned a 0 (line~\ref{alg-bfs-init} of Algorithm~\ref{alg-bfs}), but its next termination indicator $\stTermNext{\stDepth{1}}$ gets updated by $\stTermCur{\initSt}$, so that $0 < \stTermNext{\stDepth{1}} \leqslant 1$. Each new state $\stDepth{2} \in \stSetDepth{2}{1}$ generated from $\stSetDepth{1}{1}$, is ignored, since $\stTermCur{\stDepth{1}} = 0$, which is less than $\stTermToler$, and $\stDepth{2} \notin \stSet^{1}$ (line~\ref{alg-bfs-determTerm} to \ref{alg-bfs-explored} of Algorithm~\ref{alg-bfs}). Then at iteration $k=2$, the sum of termination indicators is $\AccuProbStepBFS{1}{2} = \sum_{\st \in \stSetDepth{1}{2}} \stTermCur{\st}$, where each $\stTermCur{\st}$ is a fraction of $\AccuProbStepBFS{0}{1}$, and $\AccuProbStepBFS{0}{1}= \stTermCur{\initSt} = 1$. Therefore, $\AccuProbStepBFS{1}{2}$ is solely contributed from $\AccuProbStepBFS{0}{1}$. If a self-loop transition $\{ \initSt, \react{0}, \initSt \}$ exists, then $\AccuProbStepBFS{0}{1} > \AccuProbStepBFS{1}{2}$; otherwise $\AccuProbStepBFS{0}{1} = \AccuProbStepBFS{1}{2}$. Therefore, $\AccuProbStepBFS{0}{1} \geqslant \AccuProbStepBFS{1}{2}$. Similar to the previous iteration, the updated $\stTermNext{\stDepth{2}}$ will be used in the next iteration. 
  

In general, state set $\stSetDepth{\imath}{}$ at depth $\imath$ is first obtained in iteration $\imath$ by collecting all the new states, i.e., states whose depth has not been determined, which are expanded from states in $\stSetDepth{\imath - 1}{}$. The sum of all termination indicator values for states in $\stSetDepth{\imath}{}$ is calculated at iteration $\imath + 1$ by either line~\ref{alg-bfs-update} or \ref{alg-bfs-update2} of Algorithm~\ref{alg-bfs}. To differentiate the termination indicator function $\stTermCurShort$ in different iterations, we denote $\stTermCurIter{\imath}{\st}$ as the termination indicator value for state $\st$ at iteration $\imath$. The sum of all termination indicators at iteration $\imath + 1$ is computed as follows:
\begin{align*}
  \AccuProbStepBFS{\imath}{\imath + 1} & = \sum_{\stPrime \in \stSetDepth{\imath}{\imath + 1}} \stTermCurIter{\imath+1}{\stPrime} \\
  & = \sum_{\stPrime \in \stSetDepth{\imath}{\imath + 1}}\sum_{(\st, i) \in \predStSet{\stPrime}} \left(\stTermCurIter{\imath}{\st} \cdot
    \frac{\propen{i}(\st)}{\sum_{j=1}^{m} \propen{j}(\st)}\right) \\
\end{align*}
If $\bigcup_{\stPrime \in \stSetDepth{\imath}{\imath + 1}}\predStSet{\stPrime}$ is equal to all transition firings of every state in $\stSetDepth{\imath - 1}{\imath}$, then termination indicator values for all the states at depth $\imath - 1$ are passed to depth $\imath$, and hence $\AccuProbStepBFS{\imath}{\imath + 1} =  \AccuProbStepBFS{\imath - 1}{\imath}$. On the other hand, if there exists one or more transition firings from $\stSetDepth{\imath - 1}{\imath}$ to depth other than $\imath$, then $ \AccuProbStepBFS{\imath}{\imath + 1} < \AccuProbStepBFS{\imath - 1}{\imath}$. Therefore, $\AccuProbStepBFS{\imath-1}{\imath} \geqslant \AccuProbStepBFS{\imath}{\imath + 1}$. 

We can, therefore, establish the following conclusion:
\[1 = \AccuProbStepBFS{0}{1} \geqslant \AccuProbStepBFS{1}{2} \geqslant \cdots \AccuProbStepBFS{\imath-1}{\imath} \geqslant \AccuProbStepBFS{\imath}{\imath + 1} \cdots \AccuProbStepBFS{\jmath-1}{\jmath} \geqslant \AccuProbStepBFS{\jmath}{\jmath + 1}\]
From the termination condition stated in Theorem~\ref{theorem-termination}, the slowest termination scenario, i.e., the maximal number of iterations required to terminate Algorithm~\ref{alg-top}, is the following: 
\[1 = \AccuProbStepBFS{0}{1} = \AccuProbStepBFS{1}{2} = \cdots  = \AccuProbStepBFS{\imath}{\imath+1} = \cdots =\AccuProbStepBFS{\jmath - 1}{\jmath} > \AccuProbStepBFS{\jmath}{\jmath + 1}.\]
The inequality $\AccuProbStepBFS{\jmath - 1}{\jmath} > \AccuProbStepBFS{\jmath}{\jmath + 1}$ holds only if at least one state in $\stSetDepth{\jmath - 1}{\jmath}$ executes a transition leading to a state in $\stSetDepth{0}{\jmath+1} \cup \stSetDepth{1}{\jmath+1} \cup \cdots \stSetDepth{\jmath - 1}{\jmath + 1}$, but not in $\stSetDepth{\jmath}{\jmath + 1}$. State $\st_{d+m}$ in Theorem~\ref{theorem-termination} is such a state. Additionally, the termination condition requires that at least $\AccuProbStepBFS{\imath}{\imath+1} = \cdots =\AccuProbStepBFS{\jmath - 1}{\jmath} > \AccuProbStepBFS{\jmath}{\jmath + 1}$ holds for \emph{every} depth $\jmath$. This requirement guarantees that the sum of termination indicator values keeps decreasing, with possibly many (or zero) iterations where this sum remains unchanged. Therefore, after a finite number $\xi$ of iterations, $\AccuProbStepBFS{\xi - 1}{\xi} < \stTermToler$. Since $\AccuProbStepBFS{\xi - 1}{\xi}$ is the sum of all individual termination indicator values, in the next iteration $(\xi + 1)$, termination indicator $\stTermCur{\stDepth{\xi}}$ is less than $\stTermToler$ for all states in $\stSetDepth{\xi}{\xi+1}$, and they become terminal states. Hence, $|\sg{\xi}| = |\sg{\xi+1}|$. 


Finally, Algorithm~\ref{alg-add-absorbing-state} terminates, provided its input state graph generated by Algorithm~\ref{alg-bfs} is finite. This has been proven true, and hence Algorithm~\ref{alg-add-absorbing-state} always terminates. Therefore, Theorem~\ref{alg-top} is true. \hfill $\qed$ 
\end{proof}

\section{Results}
\label{sec-results}
 
Algorithms~\ref{alg-top}, ~\ref{alg-bfs} and ~\ref{alg-add-absorbing-state} were implemented in Java as a prototype tool within the {\tt iBioSim} \emph{genetic design automation} (GDA) tool~\cite{Myers2009a,Madsen2012,Watanabe2018}. This tool constructs an approximate CTMC transition rate matrix, which is then analyzed with the help of the PRISM probabilistic model checking tool~\cite{Kwiatkowska2011}. Experiments were performed on a 3.2 GHz AMD Debian Linux PC with six cores and 64 GB of RAM. The presented CTMC approximation method was evaluated on several  CSL properties for the genetic toggle switch described in Section~\ref{sec-example}. This method is then applied to several benchmark examples, and the results are compared with those generated by the STAR tool.

\subsection{Toggle Switch}
An important metric for a toggle switch circuit is the response time. In the first set of experiments, the goal is to determine the genetic toggle switch's response time (i.e., the time it takes to switch from the OFF state to the ON state). The initial OFF state for the toggle switch has $60$ LacI, $0$ TetR, and $100$ IPTG molecules, representing the circuit has just received the set input to switch to the ON state. It should be noted that the input value of $100$ molecules is chosen to ensure that the circuit should switch to the ON state, but any moderately large value of input could be used as IPTG is represented as a boundary condition species which means that its molecule count is treated as non-depleting and that it is not consumed by any reactions that it occurs in. The CSL property, ${\tt F}(t \leqslant 2100, \text{LacI} < 20 \land \text{TetR} > 40)$, describes the probability of the circuit eventually switching to the ON state within a cell cycle of $2,100$ seconds (an approximation of the cell cycle in \emph{E. coli}~\cite{Zheng2016}). The ON state is characterized by LacI dropping below $20$ and TetR rising above $40$ molecules.

The termination indicator values are set to $10^{-5}$, $10^{-6}$, $10^{-7}$ and $10^{-9}$. Approximate state space generation and CTMC analysis are performed for each such value. In addition, intermediate verification results are generated on a time course from $0$ to $2,100$ seconds with the increment of $100$ seconds. To measure the accuracy of the presented state space approximations with different termination indicator values, a reference finite-state SCK model is created allowing both LacI and TetR to reach the upper bound of $300$ molecules each, which is significantly higher than the upper bounds of TetR and LacI for all experiments performed. The reference model, therefore, incurs significantly larger state space with $90,601$ states, but provides accurate verification results for comparison.

Both the accuracy and performance results for the response rate verification are presented in Table~\ref{tbl-resultsResponseRate}. The column ``$|\sg{}|$'' lists results
for the approximate state graph size used for the CTMC analysis, respectively. The column ``$\mathbf{\epsilon}$'' reports the difference between minimum ($\mathbf{P_{min}}$) and maximum ($\mathbf{P_{max}}$) final response rate probability, which can be taken as the uncertainty window. The columns labeled $T_{\mathbf{P_{min}}}$ and $T_{\mathbf{P_{max}}}$, provide the CTMC analysis time taken by the PRISM tool to calculate the minimum and maximum probability value, respectively.

\begin{table}[tbhp]
\begin{center}
\caption{Genetic toggle switch response rate results.}
\label{tbl-resultsResponseRate}
\vspace{2mm}
\begin{tabular*}{\textwidth}{c@{\extracolsep{\fill}}ccccccc}
\hline
$\stTermToler$ & $|\sg{}|$ & $\mathbf{P_{min}}$ & $\mathbf{P_{max}}$ & $\epsilon$ & $T_{\mathbf{P_{min}}}$ & $T_{\mathbf{P_{max}}}$ \\\hline
$\text{ref}$ & $90601$ & $0.991789007$ & $-$ & $-$ & $23.49$ & $-$ \\
$10^{-5}$   & $6171$  & $0.990640972$ & $0.991838990$ & $1.20 \times 10^{-3}$ & $0.492$ & $0.499$ \\
$10^{-6}$   & $7394$  & $0.991705919$ & $0.991794344$ & $8.84 \times 10^{-5}$ & $0.714$ & $0.629$ \\
$10^{-7}$   & $8623$  & $0.991781737$ & $0.991789578$ & $7.84 \times 10^{-6}$ & $0.811$ & $0.809$ \\
$10^{-9}$   & $11394$ & $0.991788952$ & $0.991789012$ & $5.98 \times 10^{-8}$ & $1.161$ & $1.152$ \\
\hline
\end{tabular*}
\end{center}
\end{table}

As the table shows, reducing the termination indicator value improves the accuracy of the final probability, at the price of increased performance cost. Furthermore, the final probability for $t \leqslant 2100$ of the reference model switching its state lies between the window of minimum and maximum probability for all approximate models obtained for different values of $\stTermToler$. As we decrease the value for $\stTermToler$, from $10^{-5}$ to $10^{-9}$, the error window becomes narrower from $1.20 \times 10^{-3}$ to $5.98 \times 10^{-8}$. The reference model has the final probability of $0.991789007$. With $\stTermToler=10^{-5}$, its final probability already produces very accurate final probability with significantly smaller performance cost. The approximate model only explores $6333$ states, compared to $90,601$ states from the reference model, but it produces the minimum result $0.990640972$ and maintains the error bound $1.20 \times 10^{-3}$. The runtime for the CTMC analysis on the reference model is $23.49$ seconds, much longer than the runtime for analyzing the approximate CTMC. As an additional comparison, the same toggle switch model built with pre-determined thresholds of molecule counts for LacI and TetR in~\cite{Madsen2014} produces a state graph of 70 states, and CTMC analysis with the same initial condition and CSL property reports a final probability of $98.7$ percent. The significantly smaller state space is a direct result of pre-determined thresholds, which requires prior knowledge of the circuit behavior to determine. The presented state approximation method does not require threshold determination from the user, and it achieves more accurate final probability at a slightly increased performance cost, compared to~\cite{Madsen2014}.


The second set of experiments involves computing the probability that the circuit changes state erroneously within a cell cycle of $2,100$ seconds. This behavior occurs if production of LacI erroneously and significantly inhibits TetR's production to let TetR degrade away and consequently switch state. The toggle switch is initialized to a starting state with $60$ LacI molecules, and $0$ molecules for all other species. The same CSL properties are verified and the results are summarized in Table~\ref{tbl-resultsFailureRate}.  Similar to the above experiment, the final probability for $t \leqslant 2100$ of the reference model erroneously changing its state lies between the window of minimum and maximum probability for all approximate models obtained for different values of $\stTermToler$. Decreasing the value for $\stTermToler$ from $10^{-5}$ to $10^{-9}$ decreases the error window from $4.46 \times 10^{-3}$ to $1.73 \times 10^{-7}$. Figure~\ref{fig-Error_Window_Comparison_1E-5_1E-9} shows the time-series plot for the genetic toggle switch failure rates with $\stTermToler = 10^{-5}$ and $\stTermToler = 10^{-9}$.

\begin{table}[tbhp]
\begin{center}
\caption{Genetic toggle switch failure rate results.}
\label{tbl-resultsFailureRate}
\vspace{2mm}
\begin{tabular*}{\textwidth}{c@{\extracolsep{\fill}}ccccccc}
\hline
$\stTermToler$ & $|\sg{}|$ & $\mathbf{P_{min}}$ & $\mathbf{P_{max}}$ & $\epsilon$ & $T_{\mathbf{P_{min}}}$ & $T_{\mathbf{P_{max}}}$  
\\\hline
$\text{ref}$  & $90601$ & $0.013098589$ & $-$ & $-$ & $25.041$ & $-$ \\
$10^{-5}$    & $2703$  & $0.011892475$ & $0.016356430$ & $4.46 \times 10^{-3}$ & $0.239$ & $0.241$ \\
$10^{-6}$    & $3489$  & $0.013076325$ & $0.013578975$ & $5.03 \times 10^{-4}$ & $0.287$ & $0.285$ \\
$10^{-7}$    & $4306$  & $0.013097869$ & $0.013166728$ & $6.89 \times 10^{-5}$ & $0.361$ & $0.358$ \\
$10^{-9}$    & $6697$  & $0.013098588$ & $0.013098761$ & $1.73 \times 10^{-7}$ & $0.560$ & $0.566$ \\
\hline
\end{tabular*}
\end{center}
\end{table}

\begin{figure*}[tbph]
    \centering
    \subfloat[Genetic toggle switch failure rate with $\stTermToler = 10^{-5}$\label{fig-FailureRate_ToggleSwitchAbsorb_1E-5}]{
    \includegraphics[width=\columnwidth]{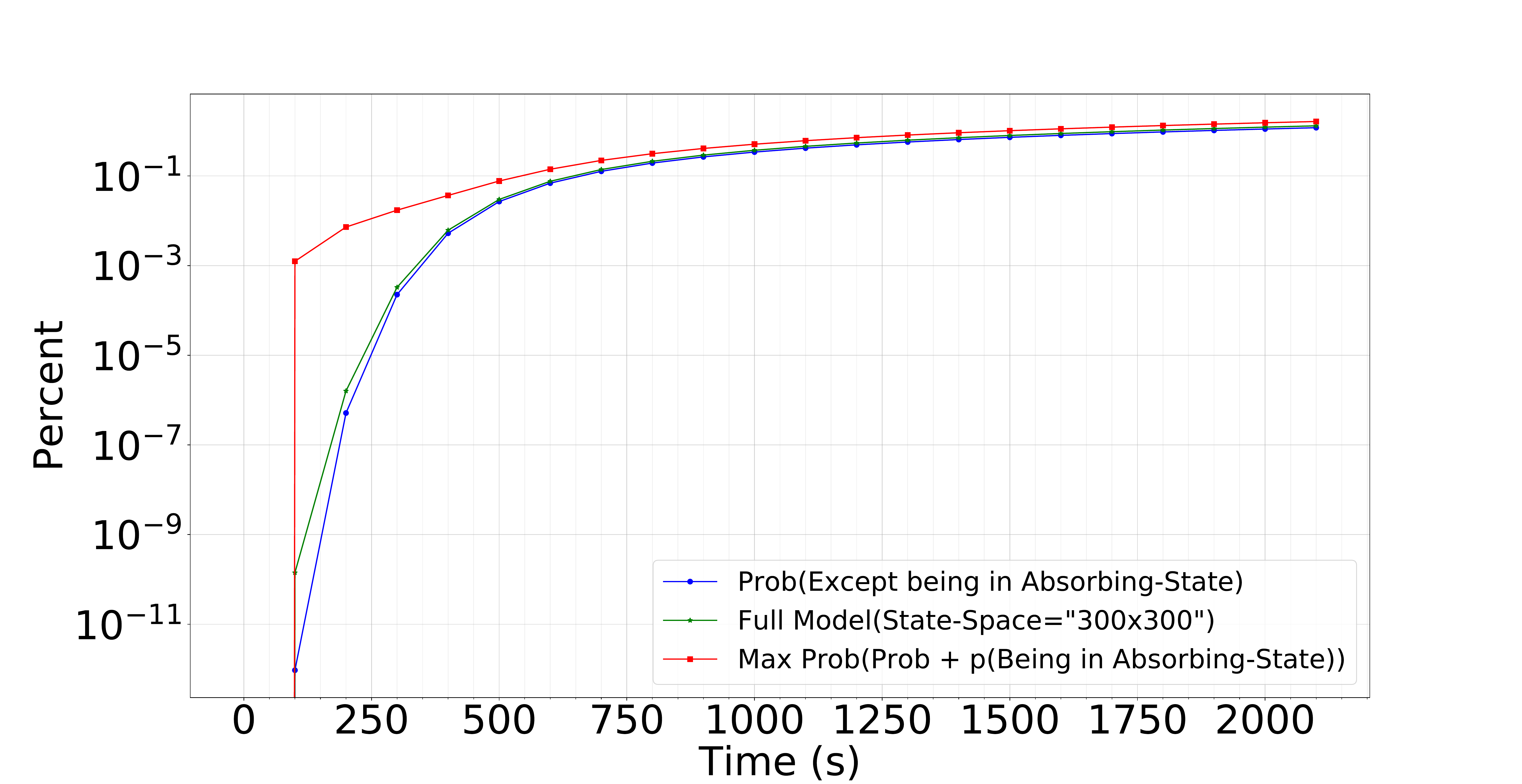}
    }\hfill
   \subfloat[Genetic toggle switch failure rate with $\stTermToler = 10^{-9}$\label{fig-Error_Window_Comparison_1E-5_1E-9}]{
     \includegraphics[width=\columnwidth]{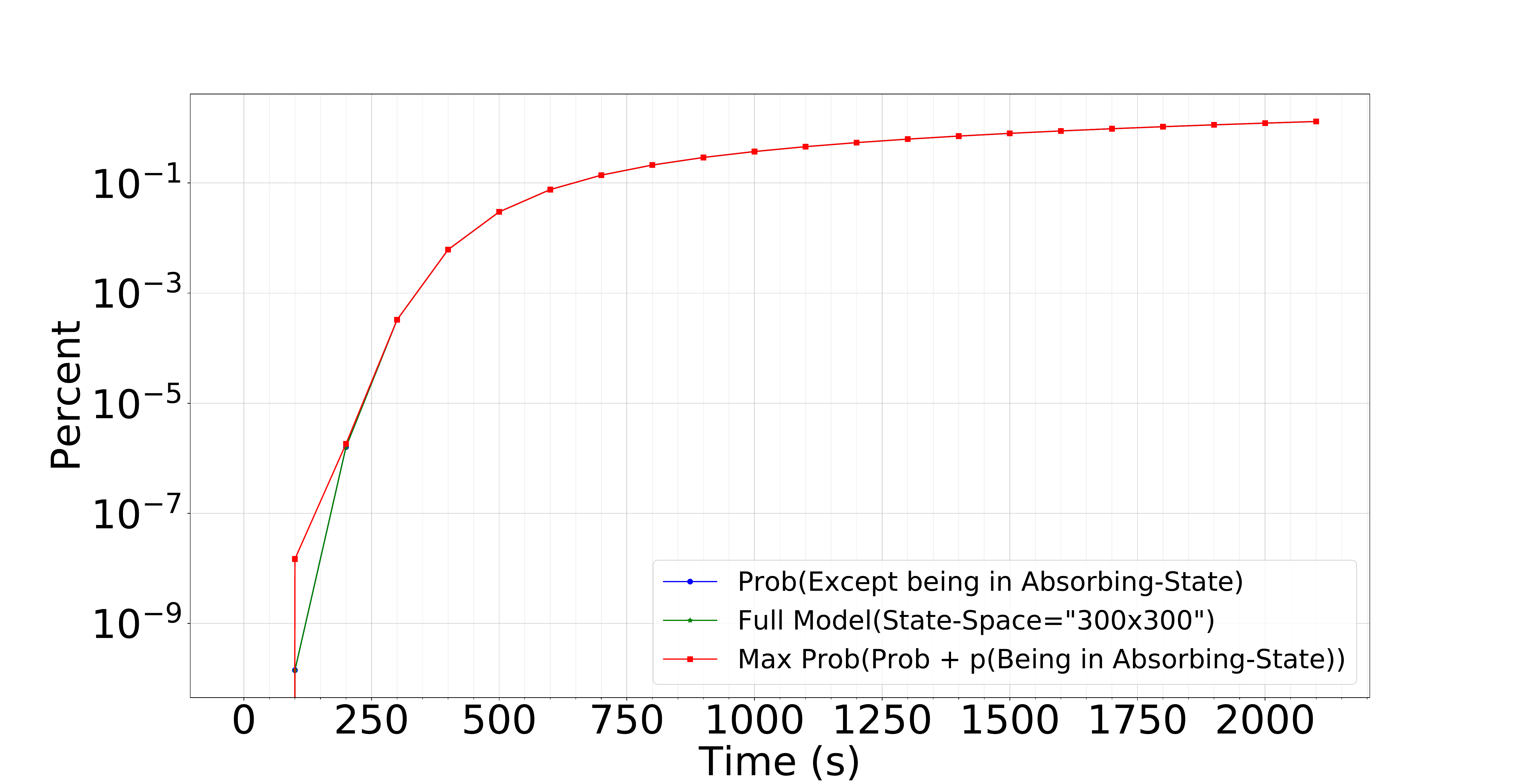}
     }%
    \caption{Error window comparison for different values of $\stTermToler$.}
  \end{figure*}


\subsection{Comparisons with the STAR Tool}
To illustrate the accuracy and efficiency of the presented method, we compared the probability distribution results with the STAR tool for the birth-death model and the presented toggle switch model. 
Table~\ref{tbl-compareBirthDeath} summarizes the comparison for a simple birth-death model, whose birth rate is 1 and death rate is 0.1. Column ``$t$'' shows the time point at which the state probability is computed, and column labeled $\epsilon_{max}$ shows the maximum absolute probability difference for the same individual state obtained from the two tools, among all explored states. Columns labeled $T_{{\tt iBioSim}}$ and $T_{{\tt PRISM}}$ list runtimes in seconds to generate the state space in {\tt iBioSim} and to analyze the model in PRISM for each given time point, respectively. Column $T_{\tt STAR}$ lists the runtime reported by STAR. The maximum probability difference reaches its peak value of $2.84 \times 10^{-6}$ at time point $t=50$. All other time points show significantly smaller errors. The run time to analyze the model in {\tt iBioSim} and PRISM is less than a second as the generated state space is only $28$ states. The STAR tool also reports a similar run time. 

Table~\ref{tbl-compareToggle} shows a comparison of results for the aforementioned toggle switch. Our proposed method produces accurate results compared to those from the STAR tool, as is indicated by the maximal probability difference ($\epsilon_{max}$).  Columns $T_{{\tt iBioSim}}$ and $T_{{\tt PRISM}}$ list runtimes in seconds to generate the state space in {\tt iBioSim} and to analyze the model in PRISM for  each given time point, respectively. Column labeled $T_{\tt STAR}$ lists the runtime reported by STAR. The combined runtime to generate the state space and analyze the model for our method is less than 24 seconds for different time points and remains almost constant as the time point $t$ increases. The STAR tool reports shorter runtime for smaller $t$ but linear increase in runtime as the time point value gets larger. 

\begin{table}[tbhp]
\caption{State probability comparison for birth-death model with $\stTermToler = 10^{-9}$.}
\label{tbl-compareBirthDeath}  
\begin{center}
\begin{tabular*}{\textwidth}{c@{\extracolsep{\fill}}cccccc}
          \hline
          $ \; t  \;$   & $\epsilon_{max}$      & $T_{{\tt iBioSim}}$ & $T_{{\tt PRISM}}$ & $T_{\tt STAR}$\\\hline
          $ \; 10 \;$   & $1.34 \times 10^{-8}$ & $0.06$ & $0.304$ & $0.22$\\
          $ \; 20 \;$   & $8.75 \times 10^{-8}$ & $0.06$ & $0.311$ & $0.34$\\
          $ \; 30 \;$   & $5.84 \times 10^{-7}$ & $0.06$ & $0.303$ & $0.46$\\
          $ \; 40 \;$   & $1.40 \times 10^{-6}$ & $0.06$ & $0.307$ & $0.59$\\
          $ \; 50 \;$   & $2.84 \times 10^{-6}$ & $0.06$ & $0.309$ & $0.72$\\
          \hline
\end{tabular*}
\end{center}
\end{table}
      
\begin{table}[tbhp]
\caption{State probability comparison for switching rate experiment of toggle-switch model with $\stTermToler = 10^{-9}$.}
\label{tbl-compareToggle}
\begin{center}
\begin{tabular*}{\textwidth}{c@{\extracolsep{\fill}}cccccc}
    \hline
          $ \; t  \;$   & $\epsilon_{max}$      & $T_{{\tt iBioSim}}$ & $T_{{\tt PRISM}}$ & $T_{\tt STAR}$\\\hline
          $ \; 400  \;$   & $1.84 \times 10^{-9}$ & $18.86$ & $4.44$ & $7.24$ \\
          $ \; 800  \;$   & $2.18 \times 10^{-9}$ & $18.86$ & $4.51$ & $17.90$\\
          $ \; 1200 \;$   & $1.91 \times 10^{-8}$ & $18.86$ & $4.59$ & $29.70$\\
          $ \; 1600 \;$   & $5.50 \times 10^{-8}$ & $18.86$ & $4.69$ & $40.65$\\
          $ \; 2000 \;$   & $1.02 \times 10^{-7}$ & $18.86$ & $4.79$ & $50.35$\\
    \hline
\end{tabular*}
\end{center}  
\end{table}

\section{Conclusion}
\label{sec-conclusion}

This chapter presents a method that builds an approximate state space of genetic circuit models to analyze infinite-state continuous-time Markov chains. This approximation method iteratively expands from the initial state using a breadth first search, computes and updates the termination indicator value for each state on-the-fly, based on the cumulative path probabilities for all incoming transitions to a state. The probability of each path segment is the ratio of the propensity of a reaction to the sum of all propensities for any given state. Our state space approximation is determined by comparing the state termination indicator to a user-provided termination threshold and only exploring states with a significant termination indicator value.
This method is capable of computing the approximate state space with no prior knowledge and is completely automated. 

For future work, we plan to improve and optimize probability approximation for re-convergent paths that close cycles during the state exploration in order to achieve potentially faster termination of the state search. We will consider different approaches to determining the termination indicator value automatically from the CSL property being analyzed. Additionally, we plan to explore augmenting our technique with bi-simulation minimization and abstraction to further minimize the generated state space and better allow for scalability. To improve performance of tool implementation, we plan to investigate tighter integration with the PRISM tool. 

\section{Acknowledgements}
The authors thank Verena Wolf for providing benchmarks and the STAR tool. The authors would also like to thank Dave Parker and Joachim Klein for providing assistance in interfacing with the PRISM tool. We also thank the reviewers for their feedback on an earlier version of this paper.

\pagebreak
\bibliographystyle{spbasic} 
\bibliography{ref} 

\end{document}